\documentclass[letterpaper, 10 pt, conference]{ieeeconf}
\usepackage{latexsym,tabularx}
\usepackage{graphicx}
\usepackage{epstopdf}
\graphicspath{{./figure/}}
\DeclareGraphicsExtensions{.pdf,.jpeg,.jpg,.png,.eps}
\usepackage{epsfig}
\usepackage{subfig}
\usepackage{psfrag,psfragx}
\usepackage{times}	%font
\usepackage{tgtermes}	%font: 
\usepackage[T1]{fontenc}	%font:
\usepackage{yfonts}
\usepackage{pgothic}
\usepackage{adjustbox}

\usepackage{amsmath,amsfonts,amssymb,array,tabulary,amsthm}
\usepackage{upgreek}
\usepackage{float}
\usepackage{arydshln}
\usepackage{algorithm,algorithmic}
\usepackage{color,url}%,hyperref}
\usepackage{framed}
\usepackage{mathdots}
\usepackage{mathrsfs} %math scripts
\usepackage{tikz}
\usetikzlibrary{decorations.pathmorphing}
\usetikzlibrary{arrows,automata,positioning,shapes}
\usetikzlibrary{decorations.pathmorphing}
\tikzset{snake it/.style={decorate, decoration=snake}}
\usepackage{pdflscape}
\usepackage{mathtools} \mathtoolsset{showonlyrefs}
\usepackage{cite}
\newtheorem{theorem}{Theorem}
\newtheorem{lemma}{Lemma}
\newtheorem{corollary}{Corollary}

\newtheorem{definition}{Definition}

\usepackage{filecontents}

\IEEEoverridecommandlockouts
\title{Herding Positive, Complex Networks}
\author{ Sebastian F. Ruf$^{1}$, Magnus Egerstedt$^{1}$, and Jeff S. Shamma$^{2}$% <-this % stops a space
	\thanks{*This work was supported by funding from King Abdullah University of Science and Technology (KAUST).}% <-this % stops a space
	\thanks{$^{1}$S.F. Ruf and M. Egerstedt are with the School of Electrical and Computer Engineering, Georgia Institute of Technology and can be reached at
		{\tt\small ruf@gatech.edu} and {\tt\small magnus@gatech.edu} resp.}%
	\thanks{$^{2}$ J.S. Shamma is with Computer, Electrical and Mathematical Science and Engineering Division, KAUST, Saudi Arabia, {\tt\small jeff.shamma@kaust.edu.sa}}%
}

\begin{document}
	\maketitle
	\thispagestyle{empty}
	\pagestyle{empty}
	\begin{abstract}
		The problem of controlling complex networks is of interest to disciplines ranging from biology to swarm robotics. However, controllability can be too strict a condition, failing to capture a range of desirable behaviors. Herdability, which describes the ability to drive a system to a specific set in the state space, was recently introduced as an alternative network control notion. This paper considers the application of herdability to the study of complex networks under the assumption that a positive system evolves on the network. The herdability of a class of networked systems is investigated and two problems related to ensuring system herdability are explored. The first is the input addition problem, which investigates which nodes in a network should receive inputs to ensure that the system is herdable. The second is a related problem of selecting the best single node from which to herd the network, in the case that a single node is guaranteed to make the system is herdable. In order to select the best herding node, a novel control energy based herdability centrality measure is introduced.
	\end{abstract}
\section{Introduction}\label{sec:intro}
 Controlling complex networks has long been of interest to the controls community \cite{siljak2011decentralized}
  and has recently received considerable attention from the complex networks community \cite{liu2015control}. Complete controllability is often used to describe the ability of a complex network to be controlled, however many systems do not require complete controllability for desired system behavior to be achieved. 
  This paper considers instead an alternative notion known as herdability, which describes systems which are not completely controllable but for which a class of desirable behaviors are still possible \cite{rufacc}. 
  
This paper considers how to apply input to a complex network to ensure that a system is herdable. Herdability is particularly applicable to understanding the behavior of complex networks. A system is completely herdable if all the elements of the state can be brought above a threshold by the application of a control input. Thresholds capture an important class of behavior in biological and social systems, in which a system reaches a tipping point and as a result the behavior of a system may change dramatically. Examples of behavior driven by thresholds include 
quorum sensing in bacteria \cite{miller2001quorum} and collective social action\cite{granovetter1978threshold,schelling1971dynamic}. 

Selecting nodes to ensure a system is herdable is an example of an input addition problem \cite{COMMAULT20133322}. Input addition problems have been previously considered in the case of controllability of complex networks. In multi-agent systems, this problem is referred to as the leader selection problem \cite{rahmani2009controllability,martini2010controllability}, which determines the controllability of a system following consensus dynamics based on a given selection of leader nodes. More broadly, the input addition problem has been considered when seeking to ensure system controllability of a system which does not necessarily follow consensus dynamics\cite{COMMAULT20133322,Liu2011,Olshevsky2014,pequito2016framework,tzoumas2016minimal}. 
The input addition problem has also been considered for the case of a structured system\cite{dion2003generic} by applying the notion of structural controllability\cite{Lin1974}. 
In the case of structured systems, it has been shown that input selection can be done efficiently \cite{Liu2011,pequito2016framework}.

In the case of known dynamics, finding the minimum number of state nodes to apply input to was shown to be NP-Hard\cite{Olshevsky2014}. A similar problem, that of selecting nodes to ensure reachability to a specific end point or subspace, was also found to be NP-hard\cite{Tzoumas2015}. In contrast to these results, this paper shows that in the case of a known, positive system it is possible to determine in linear time which nodes to apply input to in order to ensure the system is herdable. 

The second problem considered in this paper, that of characterizing nodes based on control energy has been also considered in the context of controllability. Selecting leader nodes while taking into account worst case control energy was considered in \cite{pasqualetti2014controllability}. A number of control energy centralities were introduced in \cite{summers2016submodularity}, some of which were extended to include considerations of robustness to noise in \cite{fitch2016optimal}. This paper also considers control energy; however herdability is a less stringent condition than controllability, which allows a broader array of networks to be made herdable from one node and a more comprehensive comparison to be made.
	 
	 The rest of the paper is organized as follows: In Section \ref{sec:def} the basic definitions of herdable systems are introduced. Section \ref{sec:sel} considers the problem of selecting nodes to ensure system herdability. In Section \ref{sec:cent}, a novel centrality measure is introduced to compare nodes in a herdable network. The paper concludes in Section \ref{sec:con}.
	 
	 \section{Herdable Systems} \label{sec:def}
	  A networked system can be described by its graph structure and the dynamics that act over the graph structure. Consider a graph  $\mathcal{G}=(\mathcal{V},\mathcal{E})$, where $\mathcal{E}$ is the directed edge set and $\mathcal{V}=\mathcal{V}_x\cup\mathcal{V}_u$, where $\mathcal{V}_x$ is the set of state nodes and $\mathcal{V}_u$ is the set of control nodes, which together satisfy $\mathcal{V}_x \cap \mathcal{V}_u = \emptyset$. Let $\|\mathcal{V}_u\|=m$ and $\|\mathcal{V}_x\|=n$, where $\|\cdotp\|$ denotes cardinality.  Each state node $v_{i}\in\mathcal{V}_x$ has an associated scalar state $x_{i}$ and each control node $\mu_{i}\in\mathcal{V}_u$ has a scalar input $u_{i}$. The interaction dynamics of the system are assumed to be linear: \begin{equation}\dot{x}=Ax+Bu,\end{equation} where $A\in\mathbb{R}^{n\times n}$, $B\in\mathbb{R}^{n\times m}$, $x=[x_{1},x_{2},\dots,x_{n}]^{T}$, and $u=[u_{1},u_{2}\dots,u_{n}]^{T}$. The structure of the system matrices $(A,B)$ of the linear system are derived from the underlying graph $\mathcal{G}$. In general a non-zero element of $A$, $a_{ij}$, corresponds to an edge from state node $v_{j}\in\mathcal{V}_x$ to state node $v_{i}\in\mathcal{V}_x$. Similarly a non-zero element of $B$, $b_{ij}$, corresponds to an edge from input node $u_{j}\in\mathcal{V}_u$ to state node $x_{i}\in\mathcal{V}_x$. 
	 
	 Herdability considers the general problem of going from any initial point in the state space, $x(0)$, to a terminal set.
	 Specifically, the terminal set, $\mathcal{H}_{d}$, is a shifted positive orthant, defined as $\mathcal{H}_d = \{ x \in \mathbb R^{n} : x_{i} \geq d\}$. The following definition characterizes the complete herdability of a system.
	 \begin{definition}
	 	A system is completely herdable if $\forall x(0)\in\mathbb{R}^{n}$ and $\forall d\geq0$ there exists a finite time $T$ and an input $u(t), \ t\in[0,T]$ s.t. $x(T)\in\mathcal{H}_{d}$ under control input $u(t)$.
	 \end{definition}

	 This paper makes the assumption that the dynamics evolving over the network correspond to consensus dynamics, which are an example of a positive system. 
 
	 The study of positive systems covers a large range of complex networks, including subject areas ranging from epidemic spread and, more generally, compartmental systems in biology to consensus in opinion dynamics and robotics \cite{Farina2011}.  
	 As shown in \cite{rufacc}, the herdability of a positive system can be characterized based on its underlying graph structure. 
	 
	 \begin{theorem}[Theorem 4 from \cite{rufacc}]
	 	A positive linear system is completely herdable if and only if it is input connectable, i.e. there is a path from an input to any state node in the underlying graph structure.
	 \end{theorem}
	 This paper considers the implication of the above Theorem for the application of herdability to a positive system. Note that a system is positive if the weights between nodes in a network are positive. As most complex network representations are either unweighted or have positive weights, it is reasonable to assume that the underlying dynamic is positive. 
	 
	 For a given linear system, the input connectability of the system can be checked by verifying that a directed spanning forest rooted at the control inputs covers all state nodes. Generating a directed spanning tree rooted at a node can be done in linear time via breadth-first or depth-first search\cite{kleinberg2006algorithm}, allowing system herdability to be checked efficiently.
	 
	 Let us consider in greater depth how herdability differs from controllability, specifically with regard to symmetry with respect to an input. The two major lines of work on the controllability of networks, that of structural controllability and consensus dynamics over networks, have identified symmetry with respect to a control input as an sufficient condition for the loss of controllability of a system \cite{Lin1974,rahmani2009controllability}. 

 Symmetric nodes must be controlled together, which violates the condition of complete controllability. 
 As herdability looks only at driving the state to be larger than some threshold, the herdability condition is satisfied even when the symmetric nodes are controlled to the same point. An illustrative case of symmetric systems is the star graph, shown in Fig. \ref{fig:comp}.  
	 The fact that symmetry degrades controllability explains why past analysis of controllability of complex networks has found that driver node selection avoids hubs\cite{Liu2011}. 
	 
 	 \begin{figure}
 	\begin{tabular}{cc}
 		\begin{tikzpicture}
 		\draw [-,red] (1.1768,0.8232) -- (1.8232,0.1768);
 		\draw [-] (1.1768,-0.8232) -- (1.8232,-0.1768);
 		\draw [-] (2.1768,0.1768) -- (2.8232,0.8232);
 		\draw [-,red] (2.1768,-0.1768) -- (2.8232,-0.8232);
 		\draw [->,color=red,snake it] (0,2) -- (.75,1.25);
 		\draw [->,color=blue,snake it] (0,0) -- (.75,-.75);
 		\draw [->,color=green,snake it] (2,2) -- (2.75,1.25);
 		\draw [->,draw opacity=0] (4,1) -- (0,.25);
 		\draw[red] (1,1) circle [radius=0.25];
 		\draw[blue] (1,-1) circle [radius=0.25];
 		\draw[red] (2,0) circle [radius=0.25];
 		\draw[green] (3,1) circle [radius=0.25];
 		\draw[red] (3,-1) circle [radius=0.25];
 	
 		\end{tikzpicture}&
 		\begin{tikzpicture}
 		\draw [-] (1.1768,0.8232) -- (1.8232,0.1768);
 		\draw [-] (1.1768,-0.8232) -- (1.8232,-0.1768);
 		\draw [-] (2.1768,0.1768) -- (2.8232,0.8232);
 		\draw [-] (2.1768,-0.1768) -- (2.8232,-0.8232);
 		\draw [->,color=black,snake it] (2,1.5) -- (2,.45);
 		\draw  (1,1) circle [radius=0.25];
 		\draw (1,-1) circle [radius=0.25];
 		\draw (2,0) circle [radius=0.25];
 		\draw (3,1) circle [radius=0.25];
 		\draw (3,-1) circle [radius=0.25];
 	
 		\end{tikzpicture}\\
 		Controllability & Herdability\\
 	\end{tabular}
 	\caption{Controllability analysis will select $3$ nodes in order to ensure controllability of the system. Herdability (under the assumption of a positive system) can select any node, including the middle node, as symmetry does not degrade the ability to herd the network.}
 	\label{fig:comp}
 \end{figure}
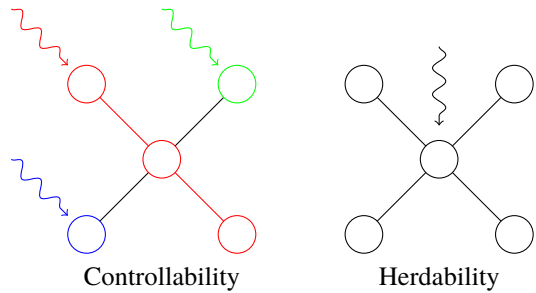

	 \section{Selecting Herding Nodes}\label{sec:sel}
	 The previous section considered the herdability of a given positive, linear system. It is often the case when interacting with networked systems that instead of being given an existing set of interconnections with input nodes, the problem is one of selecting the state nodes with which to interact to ensure the system is herdable, i.e. to design the $B$ matrix of the linear system. To this end, this section considers the input selection problem: how to select a minimal subset, $\mathbb{H}$, consisting of $N_{H}$ state nodes that ensures herdability of the system, where each element of $\mathbb{H}$ is called a herding node. Note that based on the desired terminal set, once $\mathbb{H}$ is identified, system herdability can be ensured with a $B\in\mathbb{R}^{n}$ which is $1$ at position $i$ if $i\in\mathbb{H}$ and $0$ else. 
	 
	 Consider now the problem of making a given system herdable by adding input to make a network input connectable. The solution to this problem will be called a Herding Cover. In order to generate a herding cover, the system must first be decomposed into strongly connected components (SCCs). This can be achieved in linear time by Kosaraju's algorithm\cite{kleinberg2006algorithm}. Once the strongly connected components are identified, a graph condensation is performed which generates a directed acyclic representation of the graph $\mathcal{G}$, represented as  $\mathcal{G}_{a}=(\mathcal{V}_{a},\mathcal{E}_{a})$. Each element of $\mathcal{V}_{a}$ represents a strongly connected component of $\mathcal{G}$ and an edge is in $\mathcal{E}_{a}$ if there is a link in $\mathcal{E}$ between any nodes in the respective SCCs\cite{harary2005structural}. Let $N_{r}$ be the number of roots of this acyclic representation. 
	 \begin{table*}[htbp]
	 	\centering
	 	\adjustbox{width=.87\textwidth}{
	 	\begin{tabular}{c l c c c c c c}
	 		Type &Name& N &L&Dir.& $n_{w}$ & $n_{H}$ & $n_{c}$ \\
	 		\hline

	 		Collaboration & Astro-Physics\cite{newman2001structure} & 16,706 &242,502&U& 1 & 0.062 &0.080\\
	 		&Condensed Matter Physics\cite{newman2001structure}& 16,726&95,188&U&1&0.071& 0.108 \\
	 		&Cond. Mat. Physics 2003\cite{newman2001structure} & 31,163&240,058&U&1&0.051& 0.090 \\
	 		&Cond. Mat. Physics 2005\cite{newman2001structure} & 40,421&351,384&U&1&0.045&0.083\\
	 		&High Energy Physics\cite{newman2001structure} & 8,361&31,502&U & 1 & 0.159& 0.208 \\
	 		&Network Science\cite{newman2006finding}& 1,589 &5,484&U& 1 &0.249&0.260\\
	 		&Jazz\cite{gleiser2003community}  &198 &5,484&U& 1 & 0.005&0.005\\
	 		&General Relativity\cite{leskovec2007graph}&26,196&28,980&U&1&0.813&0.827\\
	 		&&&&&&&\\
	 		Biological&C. Elegans Neural \cite{watts1998collective}&306 &2,345&D& 3.7 & 0.121 &0.190\\ 
	 		&Protein Interaction\cite{jeong2001lethality}&2,114&4,480&U&1&0.197&0.462\\
	 		&Dolphin Social \cite{lusseau2003bottlenose}& 62 &318&U& 1 & 0.016& 0.032 \\
	 		&&&&&&&\\
	 		Infrastructure&Western US Power Grid \cite{watts1998collective}& 4,941 &13,188&U& 1 & 0.0002& 0.116\\
	 		& Top Airports\cite{colizza2007reaction}&500&5960&U&1&0.002&0.250\\
	 		&Football Games\cite{girvan2002community}  & 115 &1,226&U& 1 & 0.009& 0.009\\
	 		&&&&&&&\\
	 		Online&UCIonline\cite{opsahl2009clustering}& 1,899&20,296&D & 138 & 0.291& 0.323 \\
	 		&Political Blogs\cite{adamic2005political}& 1,490&19,025&D &1.89&0.340&0.471\\
	 		&&&&&&&\\
	 		Friendship&Third Grade\cite{parker1993friendship}& 22&177&D & 1 & 0.046& 0.046\\
	 		&Fourth Grade\cite{parker1993friendship}& 24&161&D & 1 & 0.042&0.042\\
	 		&Fifth Grade\cite{parker1993friendship}&22 &103&D& 1 & 0.046&0.046\\
	 		&Highschool\cite{coleman1964introduction}&73&243&D&2&0.137&0.178\\
	 		&Fraternity\cite{bernard1979informant}&58&1,934&U&1&0.017&0.017\\
	 		&EIES 1\cite{freeman1979networkers}&32&650&D & 1 & 0.031& 0.031\\
	 		&EIES 2\cite{freeman1979networkers}&32 &759&D& 1 & 0.031& 0.031\\
	 		&Mine\cite{kapferer1969norms}&15 &88&U& 1 & 0.067 &0.067\\
	 	\end{tabular}}
	 	\caption{For each network, the table shows the number of nodes $N$, the number of edges $L$, whether the network is Undirected or Directed, the ratio of number of herding nodes to number of weakly connected components $n_{w}=\frac{N_{H}}{N_{w}}$, the fraction of herding nodes $n_{H}=\frac{N_{H}}{N}$, the fraction of driver nodes $n_{c}=\frac{N_{c}}{N}$.}
	 	\label{tab:nc_compm}
	 \end{table*}
	 
	 \begin{theorem}
	 It holds that $$N_{H}=N_{r}.$$ As such, $N_{H}$ can be determined in linear time.
	 \end{theorem}
	 \begin{proof}
	 	Consider the acyclic representation of the system graph. Each root node of this graph represents a SCC of the graph that has no in-bound edges from other SCCs. By applying input to one node from each such SCC, then by the definition of strong connectivity the entire SCC is herdable as well as all nodes downstream from the given SCC. As input is applied to all roots, the entire system is herdable. 
	 	
	 	This spanning forest representation can be computed in linear time with respect to the original network size. The roots of this forest representation can found in linear time, by checking each node in $\mathcal{N}_{a}$ to find the node with zero in-degree.
	 \end{proof}

	 \begin{corollary}\label{lem:unbound}
	 	If the graph is undirected or consists of disjoint strongly connected components, then \begin{equation}N_{H}=N_{w},\end{equation} where $N_{w}$ is the number of weakly connected components. 
	 \end{corollary}
 
 	 \begin{corollary}\label{th:one}
 	If the directed graph $\mathcal{G}$ is strongly connected, then any one node set forms the root of a Herding Cover.
 \end{corollary}

	 	 Table \ref{tab:nc_compm} shows results for analysis of the fraction of herding nodes, $n_{H}$, compared with the fraction of driver nodes, $n_{c}$, from the controllability analysis of \cite{Liu2011}. Across all considered networks $n_{H} \leq n_{c}$. In $15$ of the $24$ networks, herdability requires communication with fewer nodes than controlling the network as $n_{H}< n_{c}$. There are some networks, such as the Western US Power Grid, where $n_{H} << n_{c}$. These networks consist of a single SCC, which can be made herdable with one herding node as shown in Corollary \ref{th:one}. 
	 	 \begin{figure*}[h]
	 	\centering
	 	\begin{tabular}{cc}
	 		\begin{tikzpicture}
	 		\draw [-] (1.1768,0.8232) -- (1.8232,0.1768);
	 		\draw [-] (1.1768,-0.8232) -- (1.8232,-0.1768);
	 		\draw [-] (2.1768,0.1768) -- (2.8232,0.8232);
	 		\draw [-] (2.1768,-0.1768) -- (2.8232,-0.8232);
	 		\draw  (1,1) circle [radius=0.25];
	 		\draw (1,-1) circle [radius=0.25];
	 		\draw (2,0) circle [radius=0.25];
	 		\draw (3,1) circle [radius=0.25];
	 		\draw (3,-1) circle [radius=0.25];
	 		\node [above] at (1,0.25) {\footnotesize 0.606};
	 		\node [above] at (2,0.25) {\footnotesize 1};
	 		\node [above] at (3,0.25) {\footnotesize 0.606};
	 		\node [above] at (1,-0.75) {\footnotesize 0.606};
	 		\node [above] at (3,-0.75) {\footnotesize 0.606};
	 		\end{tikzpicture}
	 		&
	 		\includegraphics[width=0.5\linewidth]{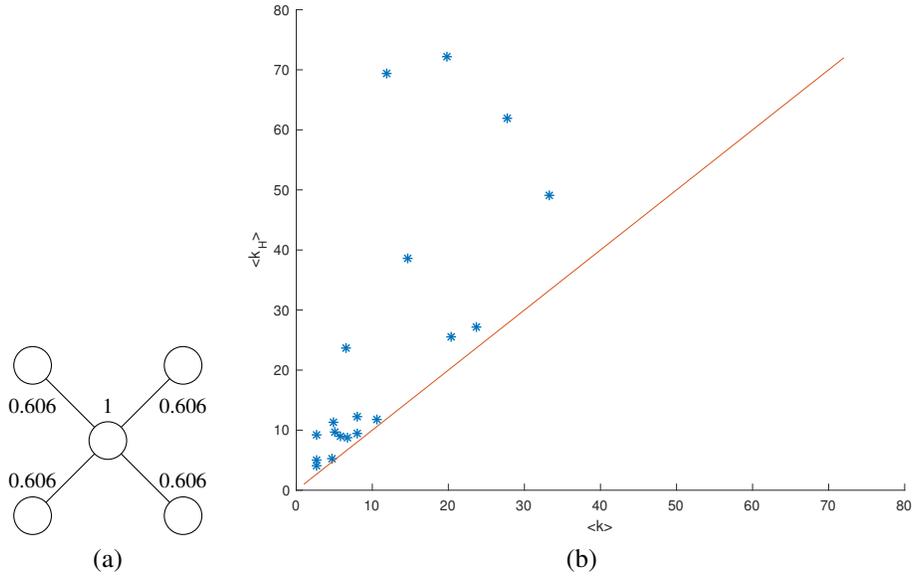}\\
	 	
	 		(a) & (b) \\
	 	\end{tabular}
	 	\caption{Herdability Centrality and Hubs: (a) Herdability centrality of a hub. (b) Plot of average degree of the complete network vs average degree of the top $10\%$ most herdable nodes, with a line representing average network degree. 
	 	}\label{fig:hubs}
	 \end{figure*}
	 \section{Herdability Centrality} \label{sec:cent}
	 
	 If the system is herdable from any one node, a secondary issue arises of selecting which one node to use as the herding node. To select between nodes in a SCC, a new herdability centrality measure is proposed which takes into account the energy required to drive the system into the set $\mathcal{H}_{d}$. 
	 While many networks are not necessarily strongly connected, as mentioned previously any directed graph can be broken down into a non-overlapping set of SCCs in linear time; allowing each SCC to be considered individually to determine the herdability centrality.
	 
	 Consider the problem of entering $\mathcal{H}_{d}$ from the origin with minimal control energy: 
	 \begin{equation}\label{eq:opt}
	 \begin{aligned}
	 J(B,d)=\underset{u(t)}{\min} &  \int_{0}^{t_{f}}\|{u(\tau)}\|^{2} d\tau \\
	 \textrm{s.t.}\  & \dot{x}(t)=Ax(t)+Bu(t), \ t\in[0,t_{f}] \\
	 &  x(t_{f})\in\mathcal{H}_{d} \\
	 & x(0)=0_{n},\\
	 \end{aligned}
	 \end{equation} 
	 where the minimum energy, $J$, is parameterized by the structure of the interaction with the control inputs, given in the matrix $B$, and by $d>0$ which is assumed to be fixed. 
	 
		 The formulation in Eq. \eqref{eq:opt} can be contrasted with the minimum energy optimal control problem as typically studied, i.e. in the context of completely controllable systems. Specifically the desired end position of the system is typically a desired final point $x_{f}$ instead of the set $\mathcal{H}$. In general, for systems that are not completely controllable, there is no guarantee that a desired $x_{f}$ or even $\mathcal{H}$ can be reached. However if the system is herdable, then by definition the reachable subspace from $0_{n}$, $R(0)$ intersects the set $\mathcal{H}_{d}$. 
		
		\begin{lemma}
			If the system is herdable, then the minimum energy to reach $\mathcal{H}_{d}$ is of the form 
			\begin{equation}
			x_{f}^{T}W_{c}^{+}x_{f},
			\end{equation} where $x_{f}\in\mathcal{H}_{d}\cap R(0)$, and $W_{c}^{+}$ is the Moore-Penrose pseudo-inverse of the controllability grammian: \begin{equation}W_c=\int_{0}^{t_f}e^{A\tau}BB^{T}e^{A^{T}\tau}d\tau.\end{equation} 
	 \end{lemma}
 \begin{proof}

	 If the network is herdable then $\exists x_{f}\in\mathcal{H}_{d}\cap R(0)$. This reachable $x_{f}$ allows the use of a number of properties of the controllability grammian. To reach $\forall x_{f}\in R(0) \cap \mathcal{H}_{d}$  requires an input $u(t)$ that satisfies $\int_{0}^{t}e^{A(t-\tau)}Bu(\tau)d\tau=x_{f}$. This $u(t)$ will have the form $u(t)=B^{T}e^{At}p$ where $W_{c}p=x_{f}$. There exists a solution to $W_{c}p=x_{f}$ as $R(0)=\mathrm{range}(W_{c})$ i.e. that $x_{f}\in\mathrm{range}(W_{c})$. These solutions are of the form \begin{equation} p^{\ast}=W^{+}_{c}x_{f}+[I-W^{+}_{c}W_{c}]x_{f}\end{equation} with $p^{\ast}=W^{+}_{c}x_{f}$ as the unique solution in the range of $W_{c}$, where $W^{+}_{c}$ can here refer to any generalized inverse\cite{james1978generalised}. If $W^{+}_{c}$ refers specifically to the Moore Penrose Inverse (or any generalized reflexive inverse) the form of the minimum energy to reach $x_{f}$ is $x_{f}W^{+}_{c}x_{f}$.  
	 \end{proof}
 
	 With the analytical expression for the minimum energy to reach $x_{f}$, it is possible to reframe Eq. \eqref{eq:opt} as the problem of choosing the optimal $x_{f}$ in the set $\mathcal{H}_{d}\cap R(0)$: 
	 
	 \begin{equation}
	 \begin{split}
	 \min_{x_{f}}\ &  x_{f}^{T}W_{c}^{+}x_{f} \\
	 \textrm{s.t.}\  &  x_f \geq d \\
	 &  x_{f} \in R(0)\\
	 & x(0)=0_{n}.\\
	 \end{split}
	 \label{eq:2}
	 \end{equation}
	 
	 Here the problem can once again be simplified further based on properties of the controllability grammian. As $W_c$ is a symmetric, real matrix, the eigenvectors of $W_c$ are mutually orthogonal and the eigenvectors with non-zero eigenvalues span the range of $W_c$ \cite{horn2012matrix}. When $\mathrm{rank}(W_{c})=r\leq n$ there are $r$ eigenvectors $\left \{v_{1},\dots,v_{r}\right\}$ associated with the $r$ non-zero eigenvalues $\lambda_{1},\dots,\lambda_{r}$ which form an orthonormal basis for $\mathrm{range}(W_{c})$. Therefore as $x_{f}\in \mathrm{range}(W_{c})$  
	 \begin{equation}
	 \label{eq:xval}
	 x_{f}=\sum_{i=1}^{r}\alpha_{i}v_{i}.
	 \end{equation} Using that $v_{i}$ are orthonormal and also eigenvectors of $W_{c}^{+}$ with associated eigenvalues $\frac{1}{\lambda_{i}}$, substituting in Eq. \eqref{eq:xval} gives
 
	 \begin{equation}
	 \begin{split}
	 \min_{\alpha}\ &  \sum_{i=1}^{r} \frac{\alpha_{i}^{2}}{\lambda_{i}}  \\
	 \textrm{s.t.}\  &  V\alpha \geq d, \\
	 \end{split}
	 \label{eq:3}
	 \end{equation}
	 where $V=\left [v_{1} \dots v_{r}\right].$ The problem in Eq.~\eqref{eq:3} can be more efficiently solved than that in Eq.~\eqref{eq:opt}, allowing larger networks to be analyzed.

	 \subsection{Calculating Herdability Centrality}
	 	 
	 \begin{figure*}[hbtp]
	 	\hspace{-1cm}
	 	\centering
	 	\begin{tabular}{ccc}
	 		\includegraphics[width=0.31\linewidth]{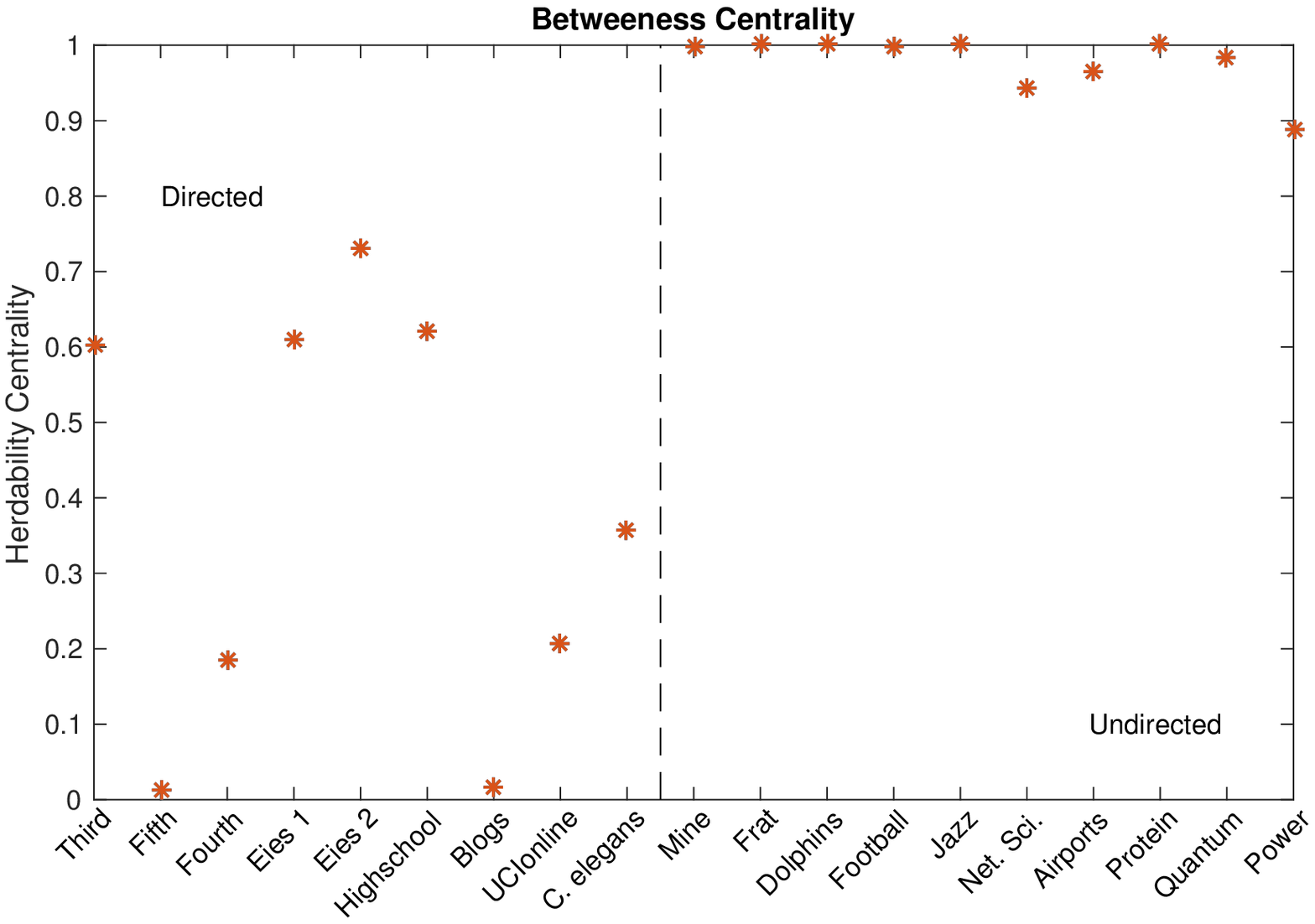}&
	 		\includegraphics[width=0.31\linewidth]{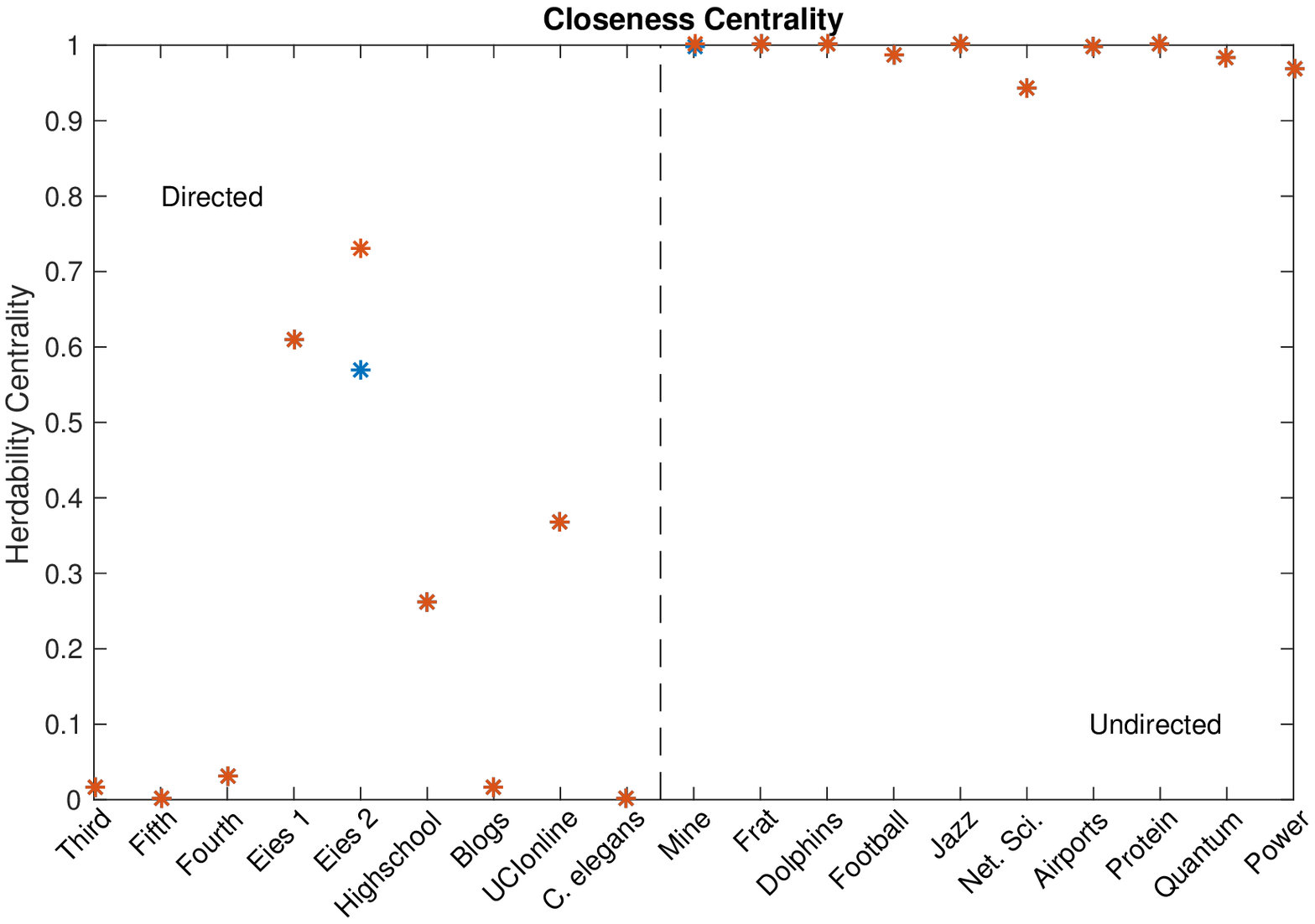}&
	 		\includegraphics[width=0.31\linewidth]{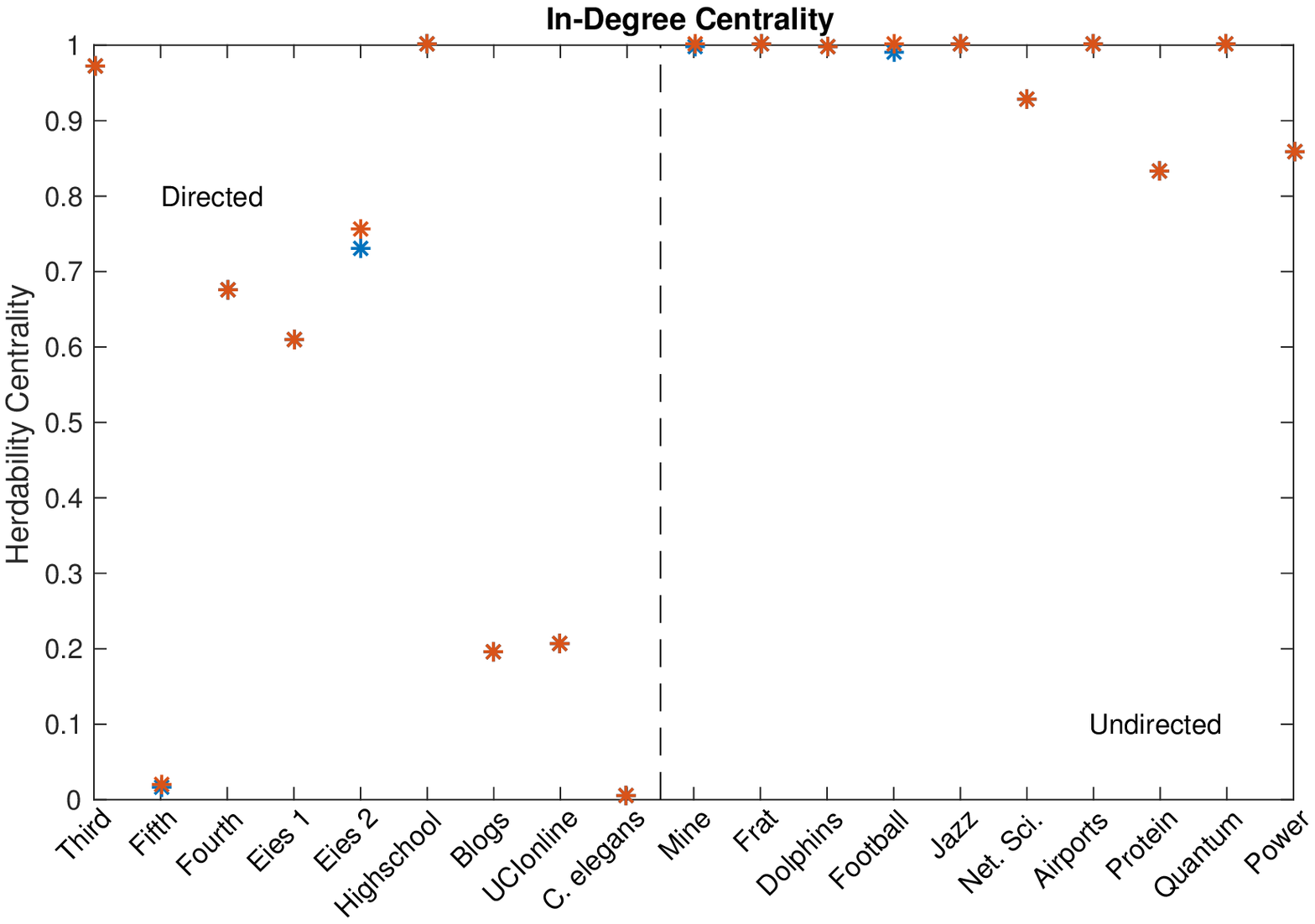}\\
	 		\includegraphics[width=0.31\linewidth]{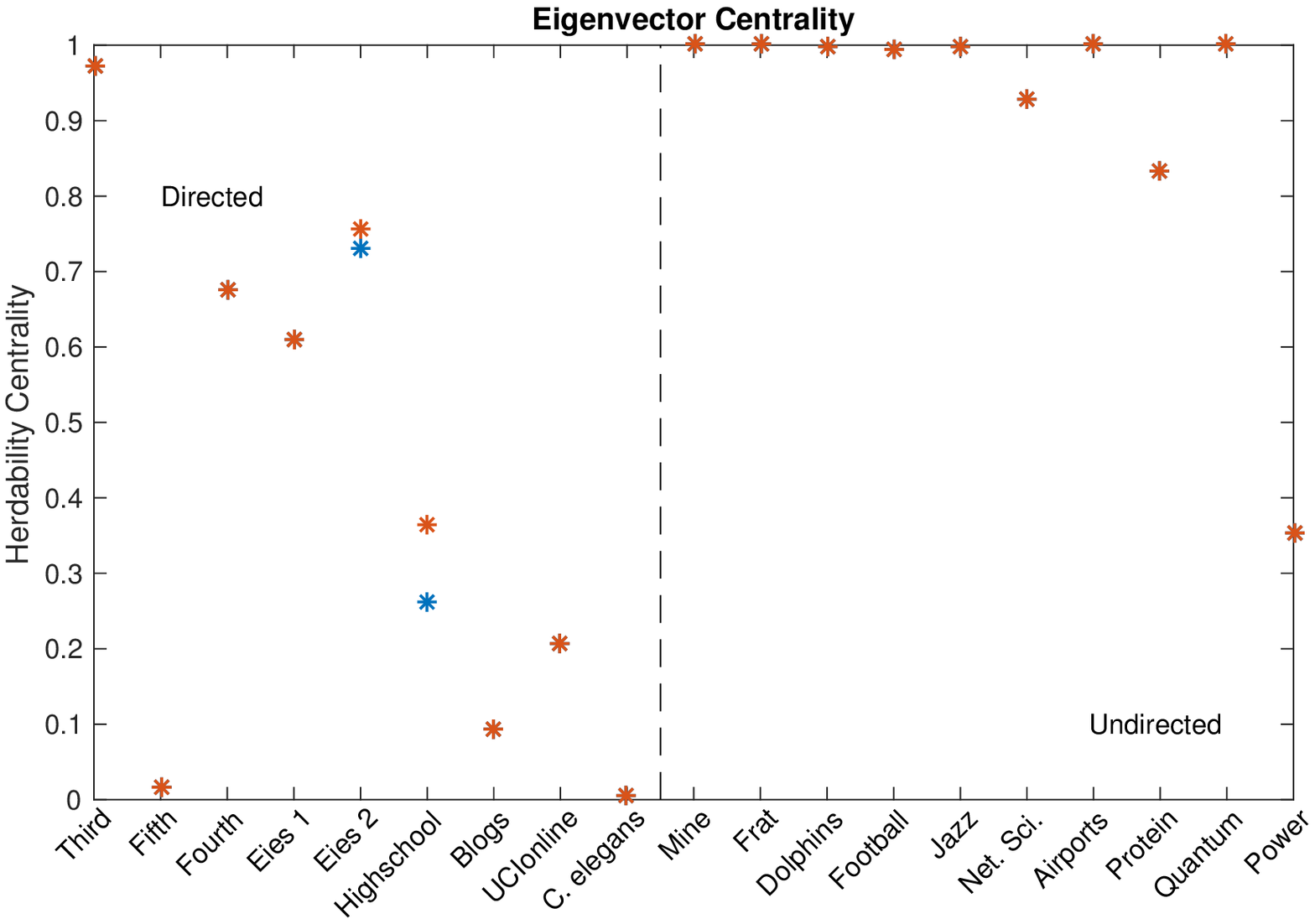}&
	 		\includegraphics[width=0.31\linewidth]{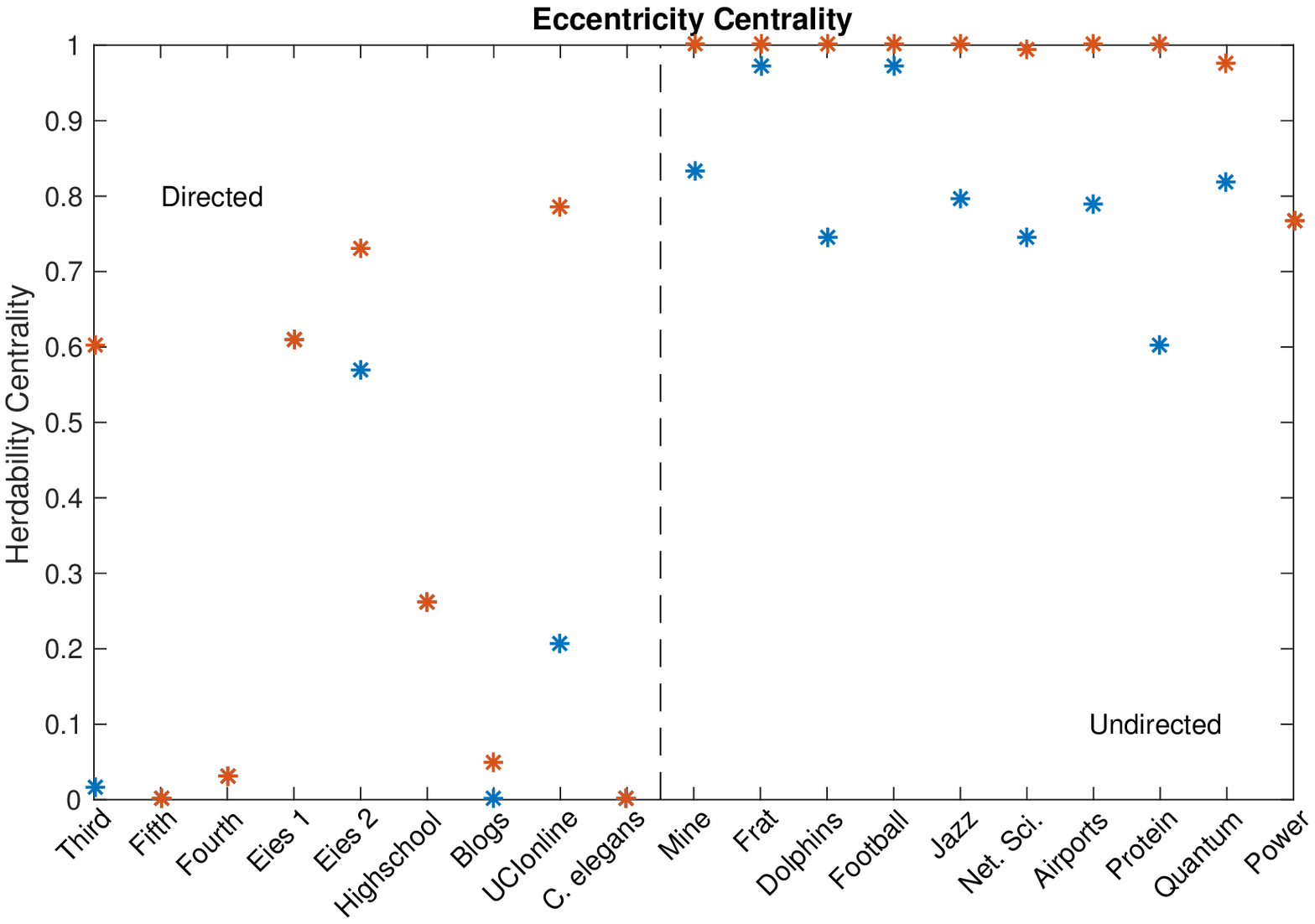}&
	 		\includegraphics[width=0.31\linewidth]{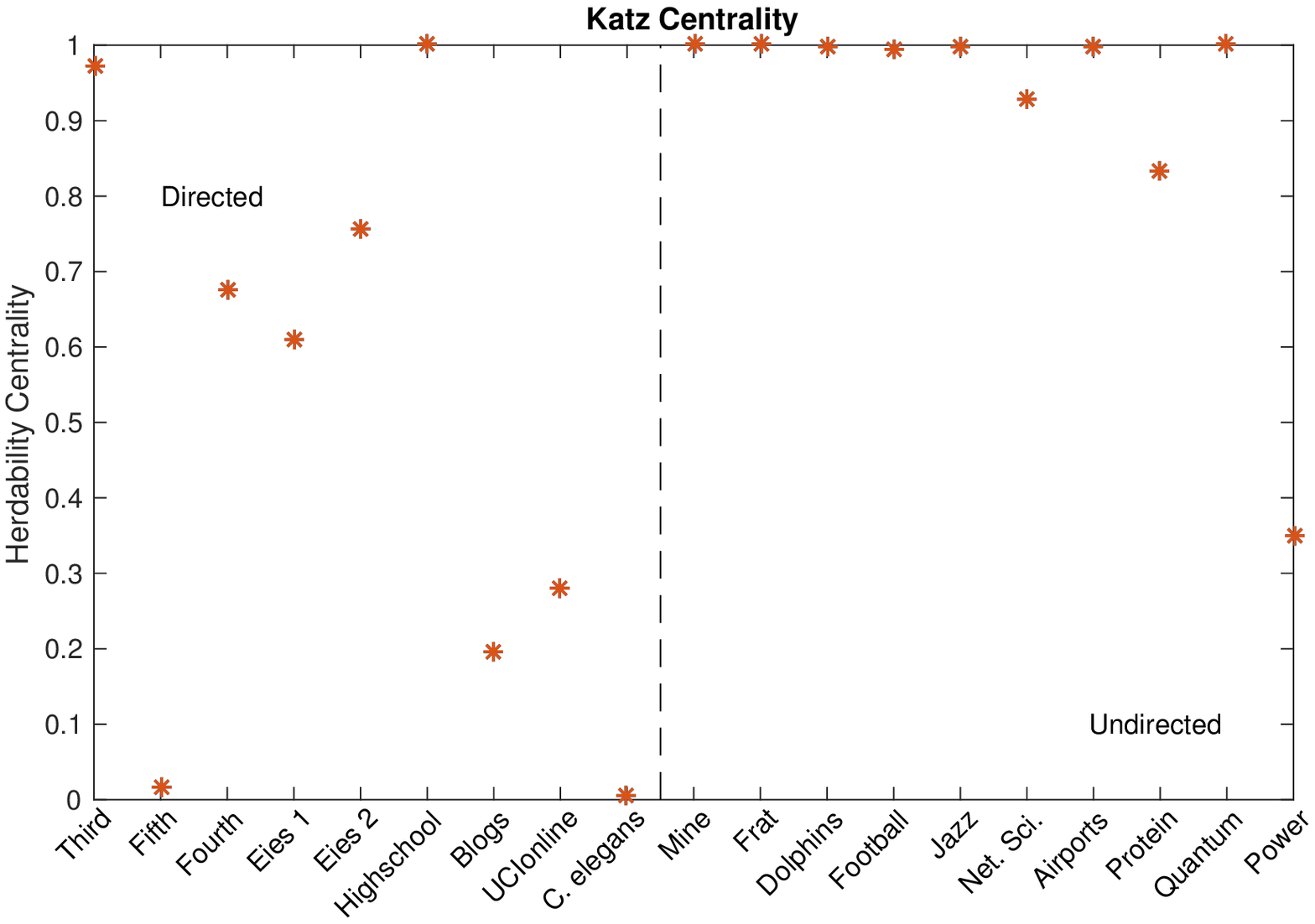}
	 		\\
	 	\end{tabular}
	 	\caption{Selecting the Highest Herdability Node: Each subgraph considers a different centrality measure and shows the highest (in red) and lowest (in blue if present) herdability centrality of the node(s) identified as having the highest value for each respective centrality. 
	 		Within each categorization (directed or undirected) the networks are organized starting with the smallest network on the left. In all undirected networks, all calculated centrality measures have high herdability centrality. In some directed networks, In-Degree and Katz centrality identify high herdability nodes. 
	 	}\label{fig:bestworst}
	 \end{figure*}
	With a simplified version of the minimum energy optimal control problem in hand, it is possible to move on to calculating herdability centrality. Each state node of the herdable system is considered in turn as the sole input node allowing the calculation of $J_{i}=J(e_{i},d)$, where $e_{i}\in\mathbb{R}^{n}$ is $1$ at position $i$ and $0$ elsewhere, and $d>0$ is fixed. The quantity $J_{i}$ is the minimum energy to reach $\mathcal{H}$ using only node $i$ as control input. In order to compare the minimum energy across nodes, the herdability centrality for node $i$, $Hc_{i}$,  is defined as \begin{equation}
	 {Hc}_{i}=\frac{\underset{k}{\min}\{J_{k}\}}{J_{i}}.
	 \end{equation}
	 Herdability centrality is normalized to be between $0$ and $1$. As reaching $\mathcal{H}$ with minimum energy is the chosen metric when interacting with these networks, the node(s) with minimum energy to reach $\mathcal{H}$ across all nodes will have the highest herdability centrality.  
	 
	 For the purpose of calculating herdability centrality of existing complex networks, the largest SCC of each considered network is used as the underlying interaction topology. The dynamics are assumed to be a modification of consensus dynamics, related to the opinion dynamic model of Taylor, which captures the effect of an external source of information on the opinion of an agent \cite{taylor1968towards}. When node $i$ is the sole herding node, the consensus dynamics are as follows:	\begin{equation}
	 \begin{split}
	 \dot{x}_{j}(t)&=\sum_{z\in\mathcal{N}_j} (x_z(t)-x_j(t)), \ \ \forall j\neq i \\
	 \dot{x}_{i}(t)&=\sum_{k\in\mathcal{N}_i} (x_k(t)-x_i(t))+u(t)-x_{i}(t),\\
	 \end{split}
	 \end{equation}
	 where $\mathcal{N}_i$ is the set of nodes with edges entering node $i$. 
	 In order to improve efficiency of the calculation, the final time is taken to be $t_{f}=\infty$ as the infinite horizon controllability grammian can be solved for efficiently, if $A$ is stable, as the solution to the continuous time Lyapunov equation $AW_c + W_c A+ BB^{T}=0.$ Note that while consensus does not normally provide a stable $A$, the model above does.

	 As mentioned previously, herdability allows hubs to be selected to herd complex systems, though it is not known a priori that hubs will indeed be selected. Fig. \ref{fig:hubs}(a) shows that the center node of the hub has the highest herdability centrality, and therefore requires the least energy to reach $\mathcal{H}_{d}$. Fig. \ref{fig:hubs}(b) shows that the introduced herdability centrality tends to select nodes that have higher than average degree, i.e. that herdability centrality tends to select hubs.

	 \subsection{Comparison to Other Centrality Measures}
	 
	 Given that herdability centrality tends to select high degree nodes, the question becomes whether it is possible to forgo the computationally expensive herdability centrality calculation in favor of an inexpensive graph structure based calculation. 
	 Table \ref{tab:cent} introduces a number of centrality measures which will be compared against herdability centrality.
	 
	 \begin{table}[htbp]
	 	\begin{tabular}{l l}
	 		Name& Description \\
	 		\hline
	 		In-Degree Centrality& The number of in-bound edges \\
	 		Eccentricity & The maximum distance from the node\\ & to any other node \\
	 		Closeness Centrality & The sum of the reciprocal of the distance \\ &to each other nodes\\
	 		Betweenness Centrality & The number of shortest paths that pass \\ & through the node divided by the total number \\ & of shortest paths between two nodes \\
	 		Eigenvalue Centrality & For node $i$, the $i$th component of \\ & the dominant eigenvector of \\ & the Adjacency Matrix \\
	 		Katz Centrality\cite{katz1953new} & The weighted sum of all paths, where a path\\ & of length $d$ receives a weight of $\alpha^{d}, \alpha>0$. 
	 	\end{tabular}
	 	\caption{Description of Centrality Measures}
	 	\label{tab:cent}
	 \end{table}
	 
	 Fig. \ref{fig:bestworst} shows that while high herdability centrality nodes tend to have high degree, the highest in-degree node does not necessarily have high herdability centrality. This holds for all centrality measures considered. 
	 In $8$ of the $19$ networks considered the traditional centrality measures overlap with the highest herdability centrality nodes. However, there is no single centrality measure which can be used reliably to select the minimum energy herding node. 
	 The overlap between herdability centrality and existing measures tends to occur in undirected networks. As control energy is related to the symmetry structure of the network \cite{martini2010controllability}, it may be that, in undirected networks, the existing centrality measures provide information about the symmetry structure. 
	 Examining the directed networks shows that size of the network seems to have no impact on overlap. 
	 For example, in the Fifth Grade Friendship network, $N=22$, all considered centrality measures select a node with low herdability centrality.

	 \section{Conclusion}\label{sec:con}
	 
	 This paper considers the application of the notion of herdability to the control of complex, positive networks. Input selection for these networks was shown to be possible in linear time. A novel centrality measure was introduced, which tends to select hubs to drive a system with minimum energy to a desired terminal set, even though hubs are not selected when considering the controllability of the system. It is shown that many centrality measures are not suitable for selecting herding nodes, especially in directed networks. The notion of herdability examines more explicitly the existing assumptions about interacting with complex networks and in doing so helps bring new insight into the control theoretic characterization of complex networks. 
	
\bibliographystyle{IEEEtran}
	\bibliography{herdability}

\end{document}